\documentclass[draftcls,onecolumn,11pt, doublespace]{IEEEtran}

\usepackage{amsmath,epsfig}
\usepackage{setspace,psfrag,cite,subfigure,amssymb, amsfonts,theorem}

\newtheorem{theorem}{Theorem}
\newtheorem{definition}{Definition}
\newtheorem{proposition}{Proposition}

\begin{document}

\title{Lorentzian Iterative Hard Thresholding: Robust Compressed Sensing with Prior Information}

\author{Rafael~E.~Carrillo
        and~Kenneth~E.~Barner

\thanks{R.E. Carrillo was with the Department
of Electrical and Computer Engineering, University of Delaware, Newark, DE 19716 USA. He is now with the Institute of Electrical Engineering, {\'E}cole Polytechnique F{\'e}d{\'e}rale de Lausanne (EPFL), CH-1015 Lausanne, Switzerland. E-mail:~rafael.carrillo@epfl.ch. K.E. Barner is with the Department of Electrical and Computer Engineering, University of Delaware, Newark, DE 19716 USA. E-mail:~barner@eecis.udel.edu.}
}%

\maketitle

\begin{abstract}
Commonly employed reconstruction algorithms in compressed sensing (CS) use the $L_2$ norm as the metric for the residual error. However, it is well-known that least squares (LS) based estimators are highly sensitive to outliers present in the measurement vector leading to a poor performance when the noise no longer follows the Gaussian assumption but, instead, is better characterized by heavier-than-Gaussian tailed distributions. In this paper, we propose a robust iterative hard Thresholding (IHT) algorithm for reconstructing sparse signals in the presence of impulsive noise. To address this problem, we use a Lorentzian cost function instead of the $L_2$ cost function employed by the traditional IHT algorithm. We also modify the algorithm to incorporate prior signal information in the recovery process. Specifically, we study the case of CS with partially known support. The proposed algorithm is a fast method with computational load comparable to the LS based IHT, whilst having the advantage of robustness against heavy-tailed impulsive noise. Sufficient conditions for stability are studied and a reconstruction error bound is derived. We also derive sufficient conditions for stable sparse signal recovery with partially known support. Theoretical analysis shows that including prior support information relaxes the conditions for successful reconstruction. Simulation results demonstrate that the Lorentzian-based IHT algorithm significantly outperform commonly employed sparse reconstruction techniques in impulsive environments, while providing comparable performance in less demanding, light-tailed environments. Numerical results also demonstrate that the partially known support inclusion improves the performance of the proposed algorithm, thereby requiring fewer samples to yield an approximate reconstruction.
\end{abstract}


\section{Introduction}
\label{sec:intro}
Compressed sensing (CS) demonstrates that a sparse, or compressible, signal can be acquired
using a low rate acquisition process that projects the signal onto a small set of vectors incoherent with the sparsity
basis~\cite{Cand08b}. There are several reconstructions methods that yield perfect or approximate reconstruction proposed in the literature (see \cite{Cand08b,Need08,Blu10} and references therein). To see a review and comparison of the most relevant algorithms see \cite{Need08}. Since noise is always present in practical acquisition systems, a range of different algorithms and methods have been developed that enable approximate reconstruction of sparse signals from noisy compressive measurements~\cite{Cand08b,Need08,Blu10}. Most such algorithms provide bounds for the $L_2$ reconstruction error based on the assumption that the corrupting noise is Gaussian, bounded, or, at a minimum, has finite variance. In contrast to the typical Gaussian assumption, heavy-tailed processes exhibit very large, or infinite, variance. Existing reconstruction algorithms operating on such processes yield estimates far from the desired original signal.

Recent works have begun to address the reconstruction of sparse signals from measurements corrupted by impulsive processes~\cite{Laska09,Carr10,Arce10a,Paredes10a,Studer12}. The works in \cite{Laska09} and \cite{Studer12} assume a sparse error and estimate both signal and error at the same stage using a modified $L_1$ minimization problem. Carrillo \emph{et al.} propose a reconstruction approach based on robust statics theory in~\cite{Carr10}. The proposed non-convex program seeks a solution that minimizes the $L_1$ norm subject to a nonlinear constraint based on the Lorentzian norm. Following this line of thought, this approach is extended in~\cite{Arce10a} to develop an iterative algorithm to solve a Lorentzian $L_0$-regularized cost function using iterative weighted myriad filters. A similar approach is used in~\cite{Paredes10a} by solving an $L_0$-regularized least absolute deviation regression problem yielding an iterative weighted median algorithm. Even though these approaches provide a robust CS framework in heavy-tailed environments, numerical algorithms to solve the proposed optimization problem are slow and complex as the dimension of the problem grows.

Recent results in CS show that modifying the recovery framework to include prior knowledge of the support improves the reconstruction results using fewer measurements~\cite{Vaswani09,Jac09}. Vaswani \emph{et. al} assume that part of the signal support is known \emph{a priori} and the problem is recast as finding the unknown support. The remainder of the signal (unknown support) is a sparser signal than the original, thereby requiring fewer samples to yield an accurate reconstruction~\cite{Vaswani09}. Although the modified CS approach in~\cite{Vaswani09} needs fewer samples to recover a signal, it employs a modified version of basis pursuit (BP)~\cite{Cand08b} to perform the reconstruction. The computational cost of solving the convex problem posed by BP can be high for large scale problems. Therefore, in~\cite{Carr10c} we proposed to extend the ideas of modified CS to iterative approaches like greedy algorithms~\cite{Need08} and iterative reweighted least squares methods~\cite{Chart08}. These algorithms construct an estimate of the signal at each iteration, and are thereby amenable to incorporation of \emph{a priori} support information (1) as an initial condition or (2) at each iteration. Although the aforementioned methods are more efficient than BP, in terms of computational cost, a disadvantage of these methods is the need to invert a linear system at each iteration.

In this paper we propose a Lorentzian based iterative hard thresholding (IHT) algorithm and a simple modification to incorporate prior signal information in the recovery process. Specifically, we study the case of CS with partially known support. The IHT algorithm is a simple iterative method that does not require matrix inversion and provides near-optimal error guarantees~\cite{Blu08,Blu09}. Hard thresholding algorithms have been previously used in image denoising~\cite{Bect04} and sparse representations~\cite{Daub04,Tropp06c}. All of these methods are particular instances of a more general class of iterative thresholding algorithms~\cite{Fig03,Elad07b}. A good general overview of iterative thresholding methods is presented in~\cite{Elad07b}. Related convergence results are also given in~\cite{Comb05}.

The proposed algorithm is a fast method with computational load comparable to the least squares (LS) based IHT, whilst having the advantage of robustness against heavy-tailed impulsive noise. Sufficient conditions for stability are studied and a reconstruction error bound is derived. We also derive sufficient conditions for stable sparse signal recovery with partially known support. Theoretical analysis shows that including prior support information relaxes the conditions for successful reconstruction. Simulations results demonstrate that the Lorentzian based IHT algorithm significantly outperform commonly employed sparse reconstruction techniques in impulsive environments, while providing comparable performance in less demanding, light-tailed environments. Numerical results also demonstrate that the partially known support inclusion improves the performance of the proposed algorithm, thereby requiring fewer samples to yield an approximate reconstruction.

The organization of the rest of the paper is as follows. Section~\ref{sec:BM} gives a brief review of CS and motivates the need of a simple robust algorithm capable of inclusion of prior support knowledge. In Section~\ref{sec:CH6LIHT} a robust iterative algorithm based on the Lorentzian norm is proposed and its properties are analyzed. In Section~\ref{sec:CH6PI} we propose simple modification for the developed algorithm to include prior signal signal information and analyze the partially known support case. Numerical experiments evaluating the performance of the proposed algorithms in different environments are presented in Section~\ref{sec:CH6Res}. Finally, we close in Section~\ref{sec:Ch6conc} with conclusions and future directions.

\section{Background and Motivation}
\label{sec:BM}

\subsection{Lorentzian Based Basis Pursuit}
\label{ssec:LBP}
Let $x\in \mathbb{R}^n$ be an $s$-sparse signal or an $s$-compressible signal. A signal is $s$-sparse if only $s$ of its coefficients are nonzero (usually $s \ll n$). A signal is $s$-compressible if its ordered set of coefficients decays rapidly and $x$ is well approximated by the first $s$ coefficients~\cite{Cand08b}.

Let $\Phi$ be an $m\times n$ sensing matrix, $m<n$, with rows that form a set of vectors incoherent with the sparsity basis \cite{Cand08b}. The signal $x$ is measured by $y=\Phi x+z$, where $z$ is the measurement (sampling) noise. It has been shown that a linear program (Basis Pursuit)  can recover the original signal, $x$, from $y$~\cite{Cand08b}. However, there are several reconstruction methods that yield perfect or approximate reconstructions proposed in the literature (see \cite{Cand08b,Need08,Chart08,Blu10} and references therein). Most CS algorithms use the $L_2$ norm as the metric for the residual error. However, it is well-known that LS based estimators are highly sensitive to outliers present in the measurement vector leading to a poor performance when the noise no longer follows the Gaussian assumption but, instead, is better characterized by heavier-than-Gaussian tailed distributions~\cite{Hube81,Arce05,Carr08,Carr10d}.

In \cite{Carr10} we propose a robust reconstruction approach coined Lorentzian basis pursuit (BP). This method is a robust algorithm capable of reconstructing sparse signals in the presence of impulsive sampling noise. We use the following non-linear optimization problem to estimate $x_0$ from $y$:
\begin{equation}\label{LLP}
\min_{x\in \mathbb{R}^n } \|x\|_{1}~\textrm{subject~to}~\|y-\Phi x\|_{LL_2,\gamma}\leq\epsilon
\end{equation}
where
\begin{equation}
\| u \|_{LL_2,\gamma}=\sum_{i=1}^{m} \log \{1+\gamma^{-2}u_i^2 \},~~u\in\mathbb{R}^m,~\gamma>0,
\end{equation}
is the Lorentzian or $LL_2$ norm. The $LL_2$ norm does not over penalize large deviations, as in the $L_2$ and $L_1$ norms cases, and is therefore a robust metric appropriate for impulsive environments~\cite{Carr10,Carr10d}. The performance analysis of the algorithm is based on the so called restricted isometry properties (RIP) of the matrix $\Phi$~\cite{Cand08b,Cand08}, which are defined in the following.
\begin{definition}
The $s$-restricted isometry constant of $\Phi$, $\delta_s$, is defined as the smallest positive quantity such that
\begin{equation*}
(1-\delta_s)\|v\|_{2}^{2}\leq \|\Phi v\|_{2}^{2}\leq (1+\delta_s)\|v\|_{2}^{2}
\end{equation*}
holds for all $v\in \Omega_s$, where $\Omega_s=\{ v\in \mathbb{R}^n | \|v\|_0\leq s\}$. A matrix $\Phi$ is said to satisfy the RIP of order $s$ if $\delta_s\in (0,1)$.
\end{definition}

Carrillo \emph{et. al} show in~\cite{Carr10} that if $\Phi$ meets the RIP of order $2s$, with $\delta_{2s}<\sqrt{2}-1$, then, for any $s$-sparse signal $x_0$ and observation noise $z$ with $\|z\|_{LL_2,\gamma}\leq \epsilon$, the solution to \eqref{LLP}, denoted as $x^*$, obeys
\begin{equation}\label{thm1bound}
\|x^*-x_0\|_2 \leq C_s \cdot2\gamma \cdot \sqrt{m(e^{\epsilon}-1)},
\end{equation}
where $C_s$ is a small constant. One remark is that $\gamma$ controls the robustness of the employed norm and $\epsilon$ the radius of the feasibility set $LL_2$ ball.

Although Lorentzian BP outperforms state of the art CS recovery algorithms in impulsive environments and achieves comparable performance in less demanding light-tailed environments, numerical algorithms to solve the optimization problem posed by Lorentzian BP are extremely slow and complex~\cite{Carr10}. Therefore, faster and simpler methods are sought to solve the sparse recovery problem in the presence of impulsive sampling noise.

\subsection{Iterative hard thresholding}
\label{ssec:IHT}
The iterative hard thresholding (IHT) algorithm is a simple iterative method that does not require matrix inversion at any point and provides near-optimal error guarantees~\cite{Blu09,Blu10}. The algorithm is described as follows.

Let $x^{(t)}$ denote the solution at iteration time $t$ and set $x^{(0)}$ to the zero vector. At each iteration $t$ the algorithm computes
\begin{equation}
x^{(t+1)}=H_s\left ( x^{(t)}+\mu\Phi^T(y-\Phi x^{(t)})  \right ),
\end{equation}
where $H_s(a)$ is the non-linear operator that sets all but the largest (in magnitude) $s$ elements of $a$ to zero and $\mu$ is a step size. If there is no unique set, a set can be selected either randomly or based on a predefined ordering. Convergence of this algorithm is proven in~\cite{Blu08} under the condition that $\|\Phi\|_{2\rightarrow 2}<1$, where $\|\Phi\|_{2\rightarrow 2}$ represents the spectral norm of $\Phi$, and a theoretical analysis for compressed sensing problems is presented in~\cite{Blu09,Blu10}.
Blumensath and Davies show in~\cite{Blu09} that if $\|z\|_2\leq \epsilon$ ($L_2$ bounded noise) and $\delta_{3s}<1/\sqrt{32}$, the reconstruction error of the IHT algorithm at iteration $t$ is bounded by
\begin{equation}
\|x-x^{(t)}\|_2\leq \alpha^t \|x\|_2 + \beta \epsilon,
\end{equation}
where $\alpha<1$ and $\beta$ are absolute constants that depend only on $\delta_{2s}$ and $\delta_{3s}$.

\subsection{Compressed sensing with partially known support}
\label{ssec:CSPKS}
Recent works show that modifying the CS framework to include prior knowledge of the support improves the reconstruction results using fewer measurements~\cite{Vaswani09,Jac09}. Let $x\in \mathbb{R}^n$ be an sparse or compressible signal in some basis $\Psi$ and denote $T=\text{supp}(x)$. In this setting, we assume that $T$ is partially known, \emph{i.e.} $T=T_0\cup\Delta$. The set $T_0\subset\{1,\ldots,n\}$ is the \emph{a priori} knowledge of the support of $x$ and $\Delta\subset\{1,\ldots,n\}$ is the unknown part of the support. This scenario is typical in many real signal processing applications, \emph{e.g.}, the lowest subband coefficients in a wavelet decomposition, which represent a low frequency approximation of the signal, or the first coefficients of a DCT transform of an image with a constant background, are known to be significant components.

The \emph{a priori} information modified CS seeks out a signal that explains the measurements and whose support contains the smallest number of new additions to $T_0$. Vaswani \emph{et al.} modify BP in \cite{Vaswani09} to find an sparse signal assuming uncorrupted measurements. This technique is extended by Jacques in \cite{Jac09} to the case of corrupted measurements and compressible signals. Jacques finds sufficient conditions in terms of RIP for stable reconstruction in this general case. The approach solves the following optimization program
\begin{equation}\label{iBPD}
\min_{x\in \mathbb{R}^n} \|x_{T_0^c}\|_{1}~~\textrm{s.~t.}~~\|y-\Phi x\|_{2}\leq\epsilon,
\end{equation}
where $x_{\Omega}$ denotes the vector $x$ with everything except the components indexed in $\Omega\subset\{1,\ldots,n\}$ set to 0.

Although the modified CS approach needs fewer samples to recover a signal, the computational cost of solving \eqref{iBPD} can be high, or complicated to implement. Therefore, we extend the ideas of modified CS to iterative approaches, such as greedy algorithms~\cite{Trop07,Need08} and iterative reweighted least squares methods~\cite{Carr09a}, in~\cite{Carr10c}. Even though the aforementioned methods are more efficient than BP, in terms of computational cost, a disadvantage is that these methods need to invert a linear system at each iteration. In the following section we develop a robust algorithm, inspired by the IHT algorithm, capable of diminishing the effect of impulsive noise and also capable of including partial support information.

\section{Lorentzian based Iterative Hard Thresholding Algorithm}
\label{sec:CH6LIHT}
In this section we propose a Lorentzian derived IHT algorithm for the recovery of sparse signals when the measurements are (possibly) corrupted by impulsive noise. First, we present the algorithm formulation and derive theoretical guarantees. Then, we describe how to optimize the algorithm parameters for enhanced performance.
\subsection{Algorithm formulation and stability guarantees}
Let $x_0\in \mathbb{R}^n$ be an $s$-sparse or $s$-compressible signal, $s < n$. Consider the sampling model
\begin{equation*}
y=\Phi x_0 +z,
\end{equation*}
where $\Phi$ is an $m\times n$ sensing matrix and $z$ denotes the sampling noise vector.
In order to estimate $x_0$ from $y$ we pose the following optimization problem:
\begin{equation}\label{Ch6P0}
\min_{x\in \mathbb{R}^n} \|y-\Phi x\|_{LL_2,\gamma}~~\text{subject to}~~\|x\|_0\leq s.
\end{equation}
However, the problem in \eqref{Ch6P0} is non-convex and combinatorial. Therefore, we derive a suboptimal strategy to estimate $x_0$ based on the gradient projection algorithm~\cite{Bertsekas99} since the Lorentzian norm is an everywhere continuous and differentiable function \cite{Carr10d}. The proposed strategy is formulated as follows. Let $x^{(t)}$ denote the solution at iteration time $t$ and set $x^{(0)}$ to the zero vector. At each iteration $t$ the algorithm computes
\begin{equation}\label{Ch6update}
x^{(t+1)}=H_s\left ( x^{(t)}+\mu g^{(t)}  \right )
\end{equation}
where $H_s(a)$ is the non-linear operator that sets all but the largest (in magnitude) $s$ elements of $a$ to zero, $\mu$ is a step size and
\begin{equation*}
g=-\nabla_{x}\|y-\Phi x\|_{LL_2,\gamma}.
\end{equation*}
If there is no unique set, a set can be selected either randomly or based on a predefined ordering. The negative gradient, $g$, can be expressed in the following form. Denote $\phi_i$ as the $i$-th row vector of $\Phi$. Then
\begin{equation}
g^{(t)}=\Phi^{T}W_{t}(y-\Phi x^{(t)})
\end{equation}
where $W_{t}$ is an $m\times m$ diagonal matrix with each element on the diagonal defined as
\begin{equation}\label{Ch6weight}
[W_{t}]_{i,i}=\frac{\gamma^2}{\gamma^2+(y_i-\phi^{T}_i x^{(t)})^2},~~i=1,\ldots,m.
\end{equation}
We coined the algorithm defined by the update in \eqref{Ch6update} Lorentzian iterative hard thresholding (LIHT). The derived algorithm is almost identical to LS based IHT in terms of computational load except for the additional cost of computing the $m$ weights in \eqref{Ch6weight} and a multiplication by an $m\times m$ diagonal matrix, with the advantage of robustness against heavy-tailed impulsive noise. Therefore the computational complexity per iteration of LIHT remains $\mathcal{O}(mn)$, which is limited by the application of the sensing operator $\Phi$ and its adjoint $\Phi^T$. If fast sensing operators are available then the computational complexity is reduced. Note that $[W_{t}]_{i,i}\leq 1$, with the weights going to zero when large deviations, compared to $\gamma$, are detected. In fact, if $W_t=I$ the algorithm reduces to the LS based IHT. Thus, the algorithm can be seen as a reweighted least squares thresholding approach, on which the weights diminish the effect of gross errors assigning a small weight for large deviations and a weight near one for deviations close to zero. Figure~\ref{Ch6fig:0b} shows an example of the obtained weight function with $\gamma=1$.

\begin{figure}[t]
\centering{ 
\includegraphics[width = 0.7\columnwidth]{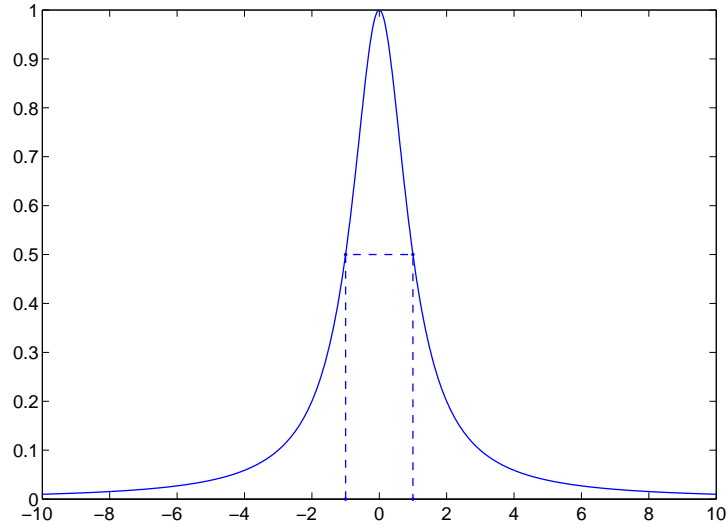}}
\caption{Weight function for $\gamma=1$. Large deviations have a weight close to zero whilst small deviations have a weight close to one.} \label{Ch6fig:0b}
\end{figure}

In the following, we show that LIHT has theoretical stability guarantees similar to those of IHT. For simplicity of the analysis we set $\mu=1$, as in \cite{Blu09}.
\begin{theorem}\label{LIHTthm1}
Let $x_0\in \mathbb{R}^n$. Define $S=\text{supp}(x_0)$, $|S|\leq s$. Suppose $\Phi \in \mathbb{R}^{m\times n}$ meets the RIP of order $3s$ and $\|\Phi\|_{2\rightarrow 2}\leq 1$. Assume $x^{(0)}=0$. Then if $\|z\|_{LL_2,\gamma}\leq \epsilon$ and $\delta_{3s}<1/\sqrt{32}$, the reconstruction error of the LIHT algorithm at iteration $t$ is bounded by
\begin{equation}
\|x_0-x^{(t)}\|_2\leq \alpha^t \|x_0\|_2 + \beta \gamma\sqrt{m(e^{\epsilon}-1)},
\end{equation}
where $\alpha=\sqrt{8}\delta_{3s}$ and $\beta=\sqrt{1+\delta_{2s}}(1-\alpha^t)(1-\alpha)^{-1}$.
\end{theorem}
Proof of Theorem \ref{LIHTthm1} follows from the fact that $W_{t}(i,i)\leq 1$, which implies that
\begin{equation*}
\|W_t z \|_2\leq \|z\|_2\leq \gamma\sqrt{m(e^{\epsilon}-1)},
\end{equation*}
where the second inequality follows from Lemma 1 in \cite{Carr10}. Argument details parallel those of the proof of Theorem~\ref{Ch6thm2} in the next section and, in fact, Theorem~\ref{LIHTthm1} is a particular case of Theorem~\ref{Ch6thm2}. Therefore we provide only a proof for the later.

Although the algorithm is not guaranteed to converge to a global minima of \eqref{Ch6P0}, it can be shown that LIHT converges to a local minima since $[W_{t}]_{i,i}\leq 1$. Thus the eigenvalues of $\Phi^T W_{t}\Phi$ are bounded above by the eigenvalues of $\Phi^T\Phi$ and the sufficient condition $\|\Phi\|_{2\rightarrow 2}\leq 1$ guarantees local convergence~\cite{Blu09}. Notice that the RIP sufficient condition for stable recovery is identical to the one required by the LS based IHT algorithm~\cite{Blu09}.

The results in Theorem~\ref{LIHTthm1} can be easily extended to compressible signals using Lemma 6.1 in~\cite{Need08}. Suppose $x_0\in \mathbb{R}^n$ is a $s$-compressible signal. Suppose $\Phi \in \mathbb{R}^{m\times n}$ meets the RIP of order $3s$ and $\|\Phi\|_{2\rightarrow 2}\leq 1$. Assume $x^{(0)}=0$. Then, if the conditions of Theorem~\ref{LIHTthm1} are met, the reconstruction error of the LIHT algorithm at iteration $t$ is bounded by
\begin{equation}
\|x_0-x^{(t)}\|_2\leq \eta \left (\|x_0-x_s\|_2+\frac{\|x_0-x_s\|_1}{\sqrt{s}} \right )+\alpha^t \|x_0\|_2 + \beta \gamma\sqrt{m(e^{\epsilon}-1)},
\end{equation}
where $\alpha=\sqrt{8}\delta_{3s}$, $\beta=\sqrt{1+\delta_{2s}}(1-\alpha^t)(1-\alpha)^{-1}$, $\eta=\sqrt{1+\delta_s}$ and $x_s$ is the best $s$-term approximation of $x_0$.

\subsection{Parameter tuning}
\label{ssec:CH6PT}
The performance of the LIHT algorithm depends on the scale parameter $\gamma$ of the Lorentzian norm and the step size, $\mu$. Therefore, we detail methods to estimate these two parameters in the following.

It is observed in \cite{Carr10} that setting $\gamma$ to half the sample range of $y$, $(y_{(1)}-y_{(0)})/2$ (where $y_{(q)}$ denotes the $q$-th quantile of $y$), often makes the Lorentzian norm a fair approximation to the $L_2$ norm. Therefore, the optimal value of $\gamma$ should be $(y'_{(1)}-y'_{(0)})/2$, where $y'=\Phi x_0$ is the uncorrupted measurement vector. Since the uncorrupted measurements are unknown, we propose to estimate the scale parameter as
\begin{equation}\label{Ch6gamma}
\gamma=\frac{y_{(0.875)}-y_{(0.125)}}{2}.
\end{equation}
This value of $\gamma$ considers implicitly a measurement vector with 25$\%$ of the samples corrupted by outliers and 75$\%$ well behaved. Experimental results show that this estimate leads to good performance in both Gaussian and impulsive environments (see Section \ref{sec:CH6Res} below).

As described in \cite{Blu10}, the convergence and performance of the LS based IHT algorithm improve if an adaptive step size, $\mu^{(t)}$, is used to normalize the gradient update. We use a similar approach in our algorithm. Let $S^{(t)}$ be the support of $x^{(t)}$ and suppose that the algorithm has identified the true support of $x_0$, \emph{i.e.} $S^{(t)}=S^{(t+1)}=S$. In this case we want to minimize $\|y-\Phi_{S}x_{S}\|_{LL_2,\gamma}$ using a gradient descent algorithm with updates of the form
\begin{equation}
x^{(t+1)}_{S}=x^{(t)}_{S}+\mu^{(t)} g^{(t)}_{S}.
\end{equation}
Finding the optimal $\mu$, \emph{i.e.}, a step size that maximally reduces the objective at each iteration, is not an easy task and in fact there is no known closed form for such an optimal step. To overcome this limitation, we propose to use the following suboptimal approach. We update the step size at each iteration as
\begin{align}
\mu^{(t)}&=\min_{\mu}\|W_t^{1/2}[y-\Phi_{S}(x^{(t)}_{S}+\mu g^{(t)}_{S})] \|_2^2\\ \nonumber
&=\frac{\|g^{(t)}_{S}\|_2^2}{\|W_t^{1/2}\Phi_{S}g^{(t)}_{S} \|_2^2},
\end{align}
which guarantees that the objective Lorentzian function is not increased at each iteration.
\begin{proposition}\label{Ch6stepupdate}
Let $\mu^{(t)}=\|g^{(t)}_{S}\|_2^2/\|W_t^{1/2}\Phi_{S}g^{(t)}_{S} \|_2^2$ and $x^{(t+1)}_{S}=x^{(t)}_{S}+\mu^{(t)} g^{(t)}_{S}$. Then, if $S^{(t)}=S^{(t+1)}=S$, the update guarantees that
\begin{equation*}
\|y-\Phi x^{(t+1)}\|_{LL_2,\gamma}\leq \|y-\Phi x^{(t)}\|_{LL_2,\gamma}.
\end{equation*}
\end{proposition}
Before proving Proposition~\ref{Ch6stepupdate}, we need a known result for square concave functions that will be used in the proof.
\begin{proposition}\label{Ch6scf}
Let $f(a)=g(a^2)$ with $g$ concave. Then for any $a,b\in \mathbb{R}$ we have the following inequality:
\begin{equation*}
f(a)-f(b)\leq \frac{f'(b)}{2b}(a^2-b^2)
\end{equation*}
which is the differential criterion for the concavity of $g$.
\end{proposition}
Now we can prove Proposition~\ref{Ch6stepupdate}.
\begin{proof}
Define
\begin{equation*}
f(a)=\log \left ( 1+\frac{a^2}{\gamma^2} \right )~~\text{and}~~r^{(t)}=y-\Phi x^{(t)}.
\end{equation*}
Using Proposition~\ref{Ch6scf} and the fact that $f(x)$ is square concave, we have the following inequality:
\begin{align*}
\sum_{i=1}^m f([r^{(t+1)}]_i)-f([r^{(t)}]_i)&\leq \frac{1}{2}\sum_{i=1}^m \frac{f'([r^{(t)}]_i)}{[r^{(t)}]_i}([r^{(t+1)}]_i^2-[r^{(t)}]_i^2)\\
&=\frac{1}{2\gamma^2}\sum_{i=1}^m[W_t]_{ii}[r^{(t+1)}]_i^2
+\frac{1}{2\gamma^2}\sum_{i=1}^m[W_t]_{ii}[r^{(t)}]_i^2.
\end{align*}
This is equivalent to
\begin{align*}
\|y-\Phi x^{(t+1)}\|_{LL_2,\gamma} &-\|y-\Phi x^{(t)}\|_{LL_2,\gamma}\\
&\leq \frac{1}{2\gamma^2}\|W_t^{1/2}(y-\Phi x^{(t+1)})\|_2^2 -\frac{1}{2\gamma^2}\|W_t^{1/2}(y-\Phi x^{(t)})\|_2^2.
\end{align*}
From the optimality of $\mu^{(t)}$ we have
\begin{equation*}
\|W_t^{1/2}(y-\Phi x^{(t+1)})\|_2^2 -\|W_t^{1/2}(y-\Phi x^{(t)})\|_2^2\leq 0.
\end{equation*}
Therefore
\begin{equation*}
\|y-\Phi x^{(t+1)}\|_{LL_2,\gamma} -\|y-\Phi x^{(t)}\|_{LL_2,\gamma}\leq 0
\end{equation*}
which is the desired result.
\end{proof}

Notably, if the support of $x^{(t+1)}$ differs from the support of $x^{(t)}$, the optimality of $\mu^{(t)}$ is no longer guaranteed. Thus, if 
\begin{equation*}
\|y-\Phi x^{(t+1)}\|_{LL_2,\gamma}> \|y-\Phi x^{(t)}\|_{LL_2,\gamma},
\end{equation*}
we use a backtracking line search strategy and reduce $\mu^{(t)}$ geometrically, \emph{i.e.}  $\mu^{(t)}\leftarrow \mu^{(t)}/2$, until the objective function in \eqref{Ch6P0} is reduced.

\section{Lorentzian Iterative Hard Thresholding with Prior Information}
\label{sec:CH6PI}
In this section we modify the LIHT algorithm to incorporate prior signal information into the recovery process. The LIHT algorithm constructs an estimate of the signal at each iteration, thereby incorporating prior knowledge at each step of the recursion. In the following we propose extensions of the LIHT algorithm to incorporate partial support knowledge. We describe then a general modification to include the model-based CS framework of~\cite{Barak10a}.

\subsection{Lorentzian iterative hard thresholding with partially known support}
\label{ssec:CH6PKS}
Let $x_0\in \mathbb{R}^n$ be an $s$-sparse or $s$-compressible signal, $s < n$. Consider the sampling model $y=\Phi x_0 +z$, where $\Phi$ is an $m\times n$ sensing matrix and $z$ denotes the sampling noise vector. Denote $T=\text{supp}(x_0)$ and assume that $T$ is partially known, \emph{i.e.} $T=T_0\cup\Delta$. Define $k=|T_0|$. We propose a simple extension of the LIHT algorithm that incorporates the partial support knowledge into the recovery process. The modification of the algorithm is described in the following.

Denote $x^{(t)}$ as the solution at iteration $t$ and set $x^{(0)}$ to the zero vector. At each iteration $t$ the algorithm computes
\begin{equation}
x^{(t+1)}=H_{s-k}^{T_0}\left ( x^{(t)}+^{(t)}\Phi^T W_t(y-\Phi x^{(t)})  \right ),
\end{equation}
where the nonlinear operator $H_u^{\Omega}(\cdot)$ is defined as
\begin{equation}
H_u^{\Omega}(a)=a_{\Omega}+H_u(a_{\Omega^{c}}),~\Omega\subset\{1,\ldots,n\}.
\end{equation}

The algorithm selects the $s-k$ largest (in magnitude) components that are not in $T_0$ and preserves all components in $T_0$ at each iteration. We coin this algorithm Lorentzian iterative hard thresholding with partially known support (LIHT-PKS).

The main result of this section, Theorem~\ref{Ch6thm2} below, shows the stability of LIHT-PKS and establish sufficient conditions for stable recovery in terms of the RIP of $\Phi$.
In the following we show that LIHT-PKS has theoretical stability guarantees similar to those of IHT~\cite{Blu09}. For simplicity of the analysis, we set $\mu=1$ as in section~\ref{sec:CH6LIHT}.

\begin{theorem}\label{Ch6thm2}
Let $x\in \mathbb{R}^n$. Define $T=\text{supp}(x)$ with $|T|=s$. Also define $T=T_0\cup \Delta$ and $|T_0|=k$. Suppose $\Phi \in \mathbb{R}^{m\times n}$ meets the RIP of order $3s-2k$ and $\|\Phi\|_{2\rightarrow 2}\leq 1$. Then if $\|z\|_{LL_2,\gamma}\leq \epsilon$ and $\delta_{3s-2k}<1/\sqrt{32}$, the reconstruction error of the IHT-PKS algorithm at iteration $t$ is bounded by
\begin{equation}
\|x_0-x^{(t)}\|_2\leq \alpha^t \|x\|_2 + \beta \gamma\sqrt{m(e^{\epsilon}-1)},
\end{equation}
where
\begin{equation*}
\alpha=\sqrt{8}\delta_{3s-2k}~~\text{and}~~\beta=\sqrt{1+\delta_{2s-k}}\left ( \frac{1-\alpha^t}{1-\alpha}\right ).
\end{equation*}
\end{theorem}
\begin{proof}
Suppose $x\in\mathbb{R}^n$ and $T=\mbox{supp}(x)$, $|T|=s$ ($s$-sparse signal). If $T=T_0\cup \Delta$, then $|\Delta|=s-k$ where $|T_0|=k$. Define
\begin{equation}\label{gradient}
a^{(t)}=x^{(t)}+\Phi^T W_t(y-\Phi x^{(t)}).
\end{equation}
The update at each iteration $t+1$ can be expressed as:
\begin{equation}\label{update}
x^{(t+1)}=a^{(t)}_{T_0}+H_{s-k}(a^{(t)}_{T^c_0}),
\end{equation}
therefore $x^{(t+1)}_{T_0} = a^{(t)}_{T_0}$. The residual (reconstruction error) at iteration $t$ is defined as $r^{(t)}=x-x^{(t)}$.

Define $T^{(t)}=\text{supp}(x^{(t)})$ and $U^{(t)}=\text{supp}\left(H_{s-k}(a^{(t)}_{T^c_0})\right)$. It can be easily checked for all $t$ that $|\text{supp}(a^{(t)}_{T_0})|= k$, $|U^{(t)}| = s-k$ and $|T^{(t)}| = s$. Also define
\begin{equation*}
B^{(t+1)} = T \cup T^{(t+1)}= T_0 \cup \Delta \cup U^{(t+1)}.
\end{equation*}
Then, the cardinality of the set $B^{(t+1)}$ is upper bounded by
\begin{equation*}
|B^{(t+1)}|\leq |T_0|+ |\Delta|+ |U^{(t+1)}|=2s-k.
\end{equation*}
The error $r^{(t+1)}$ is supported on $B^{(t+1)}$. Using the triangle
inequality we have
\begin{equation*}
\|x_{B^{(t+1)}}-x^{(t+1)}_{B^{(t+1)}} \|_2 \leq \|x_{B^{(t+1)}}-a^{(t)}_{B^{(t+1)}}\|_2
+\|x^{(t+1)}_{B^{(t+1)}}-a^{(t)}_{B^{(t+1)}} \|_2.
\end{equation*}
We start by bounding $\|x^{(t+1)}_{B^{(t+1)}}-a^{(t)}_{B^{(t+1)}}
\|_2$. Remember that $x^{(t+1)}_{T_0} = a^{(t)}_{T_0}$ and that by definition of the thresholding operator, $x^{(t+1)}_{T^c_0}$ is
the best ($s-k$)-term approximation to $a^{(t)}_{T^c_0}$. Thus, $x^{(t+1)}$ is closer to $a^{(t)}$ than $x$, on $B^{(t+1)}$, and we have
\begin{equation*}
\|x^{(t+1)}_{B^{(t+1)}}-a^{(t)}_{B^{(t+1)}}\|_2  \leq \|x_{B^{(t+1)}}-a^{(t)}_{B^{(t+1)}}\|_2.
\end{equation*}
Therefore the error at iteration $t+1$ is bounded by
\begin{equation*}
\|x_{B^{(t+1)}}- x^{(t+1)}_{B^{(t+1)}}  \|_2 \leq 2
\|x_{B^{(t+1)}}-a^{(t)}_{B^{(t+1)}} \|_2.
\end{equation*}

Rewrite \eqref{gradient} as
\begin{equation*}
a^{(t)}= x^{(t)}+\Phi^T W_t\Phi x- \Phi^T W_t\Phi x^{(t)}+ \Phi^T W_t z.
\end{equation*}
Denote $\Phi_{\Omega}$ as the submatrix obtained by selecting the columns indicated by $\Omega$. Then
\begin{equation*}
a^{(t)}_{B^{(t+1)}}= x^{(t)}_{B^{(t+1)}}+\Phi^T_{B^{(t+1)}} W_t\Phi r^{(t)}+ \Phi^T_{B^{(t+1)}}W_t z
\end{equation*}
and we can bound the estimation error as
\begin{align*}
\|x_{B^{(t+1)}}- x^{(t+1)}_{B^{(t+1)}} \|_2 & \leq 2\|x_{B^{(t+1)}}-x^{(t)}_{B^{(t+1)}}-\Phi^T_{B^{(t+1)}}W_t\Phi r^{(t)}-\Phi^T_{B^{(t+1)}}W_t z \|_2 \\
& \leq 2 \|r^{(t)}_{B^{(t+1)}}-\Phi^T_{B^{(t+1)}}W_t\Phi r^{(t)}\|_2 + 2\|\Phi^T_{B^{(t+1)}}W_t z \|_2 \\
& \leq 2 \|(I-\Phi^T_{B^{(t+1)}}W_t\Phi_{B^{(t+1)}})r^{(t)}_{B^{(t+1)}}- \Phi^T_{B^{(t+1)}}W_t\Phi_{B^{(t)} \backslash B^{(t+1)}}r^{(t)}_{B^{(t)} \backslash B^{(t+1)}}  \|_2 \\
&+ 2\|\Phi^T_{B^{(t+1)}} W_t z \|_2 \\
& \leq 2 \|(I-\Phi^T_{B^{(t+1)}}W_t\Phi_{B^{(t+1)}})r^{(t)}_{B^{(t+1)}}\|_2\\
&+ 2 \|  \Phi^T_{B^{(t+1)}}W_t\Phi_{B^{(t)} \backslash B^{(t+1)}}r^{(t)}_{B^{(t)} \backslash B^{(t+1)}}  \|_2+ 2\|\Phi^T_{B^{(t+1)}}W_t z \|_2.
\end{align*}
Since $[W_{t}]_{i,i}\leq 1$ the eigenvalues of $\Phi^T W_{t}\Phi$ are bounded above by the eigenvalues of $\Phi^T\Phi$, and, therefore,
\begin{align*}
\|x_{B^{(t+1)}}- x^{(t+1)}_{B^{(t+1)}} \|_2 & \leq 2 \|(\Phi^T_{B^{(t+1)}}\Phi_{B^{(t+1)}}-I)\|_{2\rightarrow 2}\|r^{(t)}_{B^{(t+1)}}\|_2\\
&+ 2 \|  \Phi^T_{B^{(t+1)}}\Phi_{B^{(t)} \backslash B^{(t+1)}}\|_{2\rightarrow 2}\|r^{(t)}_{B^{(t)} \backslash B^{(t+1)}}  \|_2+ 2\|\Phi^T_{B^{(t+1)}}W_t z \|_2.
\end{align*}
Notice that
\begin{align*}
|B^{(t)} \cup B^{(t+1)}| & = |T_0 \cup \Delta \cup U^{(t+1)} \cup U^{(t)}|\\
& \leq |T_0| + |\Delta| + 2|U^{(t)}|= 3s-2k.
\end{align*}
Using basic properties of the restricted isometry constants (see Lemma 1 from \cite{Blu09}) and the fact that $\delta_{3s-2k}>\delta_{2s-k}$ we have the following. Define $\eta=2\sqrt{1+\delta_{2s-k}}$.
\begin{align*}
\|x_{B^{(t+1)}} - x^{(t+1)}_{B^{(t+1)}}  \|_2 &\leq 2\delta_{2s-k} \| r^{(t)}_{B^{(t+1)}} \|_2 + 2\delta_{3s-2k}\| r^{(t)}_{B^{(t)} \setminus B^{(t+1)}} \|_2  +  \eta\|W_t z\|_2\\
& \leq 2\delta_{3s-2k} \left(\| r^{(t)}_{B^{(t+1)}} \|_2 + \| r^{(t)}_{B^{(t)}
\setminus B^{(t+1)}} \|_2 \right) + \eta\| W_t z\|_2.
\end{align*}

Since $ B^{(t)} \backslash B^{(t+1)}$ and $B^{(t+1)}$ are disjoint sets we have $\| r^{(t)}_{B^{(t+1)}} \|_2 + \| r^{(t)}_{B^{(t)} \setminus B^{(t+1)}} \|_2 \leq \sqrt{2}\|r^{(t)}_{B^{(t)} \cup B^{(t+1)}} \|_2$. Thus, the estimation error at iteration $t+1$ is bounden by
\begin{equation*}
\| r^{(t+1)} \|_2  \leq \sqrt{8}\delta_{3s-2k}\| r^{(t)} \|_2 +  \eta\|W_t z\|_2.
\end{equation*}
This is a recursive error bound. Define $\alpha=\sqrt{8}\delta_{3s-2k}$ and assume $x^{(0)}=0$. Then
\begin{equation}\label{series}
\| r^{(t)} \|_2  \leq \alpha^t \|x\|_2 + \eta\|W_t z\|_2
\sum_{j=0}^t \alpha^j.
\end{equation}
We need $\alpha=\sqrt{8}\delta_{3s-2k} < 1$ for the series in \eqref{series} to
converge. For faster convergence and better stability we restrict $\sqrt{8}\delta_{3s-2k} < 1/2$, which yields the sufficient condition in Theorem~\ref{Ch6thm2}. Now we just need to bound $\|z\|_2$. Note that $[W_{t}]_{i,i}\leq 1$, which implies that
\begin{equation*}
\|W_t z \|_2\leq \|z\|_2\leq \gamma\sqrt{m(e^{\epsilon}-1)},
\end{equation*}
where the second inequality follows from Lemma 1 in \cite{Carr10}. 
\end{proof}

A sufficient condition for stable recovery of the LIHT algorithm is $\delta_{3s}<1/\sqrt{32}$ (see section~\ref{sec:CH6LIHT}), which is a stronger condition than that required by LIHT-PKS, since $\delta_{3s-2k}<\delta_{3s}$. Having a RIP of smaller order means that $\Phi$ requires fewer rows to meet the condition, \emph{i.e.}, fewer samples to achieve approximate reconstruction. Notice that when $k=0$ (cardinality of the partially known support), we have the same condition required by LIHT. The results in Theorem~\ref{Ch6thm2} can be easily extended to compressible signals using Lemma 6.1 in~\cite{Need08}, as was done in the previous section for LIHT.

\subsection{Extension of Lorentzian iterative hard thresholding to model-sparse signals}
\label{ssec:CH6MB}
Baraniuk \emph{et. al} introduced a model-based CS theory that reduces the degrees of freedom of a sparse or compressible signal~\cite{Duarte09a,Barak10a}. The key ingredient of this approach is to use a more realistic signal model that goes beyond simple sparsity by codifying the inter-dependency structure among the signal coefficients. This signal model might be be a wavelet tree, block sparsity or in general a union of $s$-dimensional subspaces~\cite{Barak10a}.

Suppose $\mathcal{M}_{s}$ is a signal model as defined in~\cite{Barak10a} and also suppose that $x_0 \in\mathcal{M}_{s}$ is an $s$-model sparse signal. Then, a model-based extension of the LIHT algorithm is motivated by solving the problem
\begin{equation}\label{Ch6PM}
\min_{x\in \mathcal{M}_{s}} \|y-\Phi x\|_{LL_2,\gamma},
\end{equation}
using the following recursion:
\begin{equation}
x^{(t+1)}=\mathbb{M}_{s}\left ( x^{(t)}+\mu^{(t)}\Phi^T W_t(y-\Phi x^{(t)})  \right ),
\end{equation}
where $\mathbb{M}_{s}(a)$ is the best $s$-term model-based operator that projects the vector $a$ onto $\mathcal{M}_{s}$. One remark to make is that, under the model-based CS framework of \cite{Barak10a}, this prior knowledge model can be leveraged in recovery with the resulting algorithm being similar to LIHT-PKS.

\section{Experimental Results}
\label{sec:CH6Res}
\subsection{Robust Reconstruction: LIHT}
Numerical experiments that illustrate the effectiveness of the LIHT algorithm are presented in this section. All experiments utilize synthetic $s$-sparse signals in a Hadamard basis, with $s=8$ and $n=1024$. The nonzero coefficients have equal amplitude, equiprobable sign, randomly chosen position, and average power fixed to 0.78. Gaussian sensing matrices are employed with $m=128$. One thousand  repetitions of each experiment are averaged and reconstruction SNR is used as the performance measure. Weighted median regression (WMR)~\cite{Paredes10a} and LS-IHT~\cite{Blu10} are used as benchmarks.

To test the robustness of the methods, we use two noise models: $\alpha$-stable distributed noise and Gaussian noise plus gross sparse errors. The Gaussian noise plus gross sparse errors model is referred to as  contaminated $p$-Gaussian noise for the remainder of the paper, as $p$ represents the amount of gross error contamination. To validate the estimate of $\gamma$ discussed in Section~\ref{ssec:CH6PT} we make a comparison between the performance of LIHT equipped with the optimal $\gamma$, denoted as LIHT-$\gamma_1$, and the signal-estimated $\gamma$, denoted as LHIT-$\gamma_2$. The optimal $\gamma$ is set as half the sample range of the clean measurements.

For the first experiment we consider a mixed noise environment, using contaminated $p$-Gaussian noise. We set the Gaussian component variance to $\sigma^2=10^{-2}$, resulting in an SNR of 18.9321 dB when $p=0$. The amplitude of the outliers is set as $\delta=10^{3}$ and $p$ is varied from $10^{-3}$ to $0.5$. The results are shown in Figure~\ref{Ch6fig:1} (a). The results demonstrate that LIHT outperforms WMR and IHT. Moreover, the results also demonstrate the validity of the estimated $\gamma$. Although the reconstruction quality achieved by LIHT-$\gamma_2$ is lower than that achieved LIHT-$\gamma_1$, the SNR of LIHT-$\gamma_2$ is greater than 20~dB for a broad range of contamination factors $p$, including contaminations up to 5\% of the measurements.

The second experiment explores the behavior of LIHT in very impulsive environments. We compare again against IHT and WMR, this time with $\alpha$-Stable sampling noise. The scale parameter of the noise is set as $\sigma=0.1$ for all cases and the tail parameter, $\alpha$, is varied from 0.2 to 2, \emph{i.e.}, very impulsive to the Gaussian case, Figure~\ref{Ch6fig:1} (b). For small values of $\alpha$, all methods perform poorly, with LIHT yielding the most acceptable results. Beyond $\alpha=0.6$, LIHT produces faithful reconstructions with a SNR greater than 20~dB, and often 10 dB greater than IHT and WMR results. Notice that when $\alpha=2$ (Gaussian case) the performance of LIHT is comparable with that of IHT, which is least squares based. Also of note is that the SNRs achieved by LIHT-$\gamma_1$ and LIHT-$\gamma_2$ are almost identical, with LIHT-$\gamma_1$ slightly better.

\begin{figure}[p]

\begin{minipage}[b]{\columnwidth}
  \centering
  \centerline{\epsfig{figure=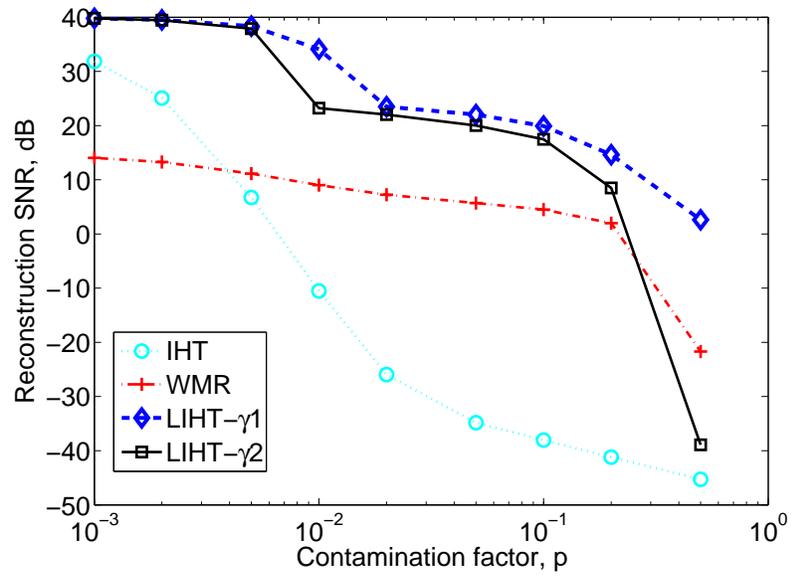,width=0.7\columnwidth}}
  \centerline{(a) }\medskip
\end{minipage}
\hfill
\begin{minipage}[b]{\columnwidth}
  \centering
  \centerline{\epsfig{figure=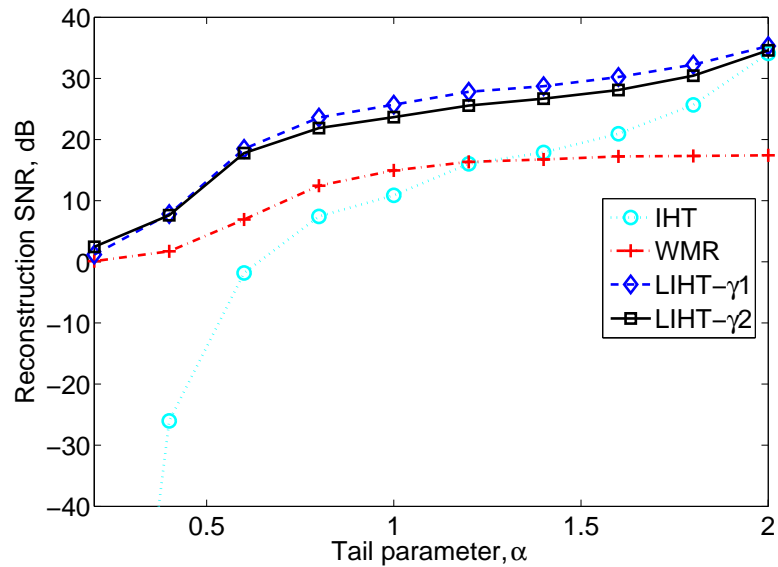,width=0.7\columnwidth}}
  \centerline{(b) }\medskip
\end{minipage}

\caption{Comparison of LIHT with LS-IHT and WMR for impulsive contaminated samples, $s=8$, $n=1024$ and $m=128$. (a) Contaminated p-Gaussian, $\sigma^2=0.01$. R-SNR as a function of the contamination parameter, p. (b) $\alpha$-stable noise, $\sigma=0.1$. R-SNR as a function of the tail parameter, $\alpha$. }
\label{Ch6fig:1}
\end{figure}

For the next experiment, we evaluate the performance of LIHT as the number of measurements varies for different levels of impulsiveness. The number of measurements is varied from 16 (twice the sparsity level) to 512 (half the dimension of $x_0$). The sampling noise model used is $\alpha$-stable with four values of $\alpha$: 0.5, 1,1.5, 2. The results are summarized in Figure~\ref{Ch6fig:2}, which show that, for $\alpha\in[1,2]$, LIHT yields fair reconstructions from 96 samples. However for $\alpha=0.5$ (most impulsive case of the four), more samples are needed, 256, to yield a fair reconstruction. Results of IHT with Gaussian noise ($\alpha=2$) are also included for comparison. Notice that the performance of LIHT is comparable to that of IHT for the Gaussian case. One remark is that LIHT needs more measurements, for a fixed sparsity level, than Lorentzian BP to yield an accurate reconstruction (see results in \cite{Carr10}). This is a general disadvantage of thresholding algorithms over $L_1$ minimization based methods~\cite{Blu09}.

\begin{figure}[t]
\centering{ 
\includegraphics[width = 0.7\columnwidth]{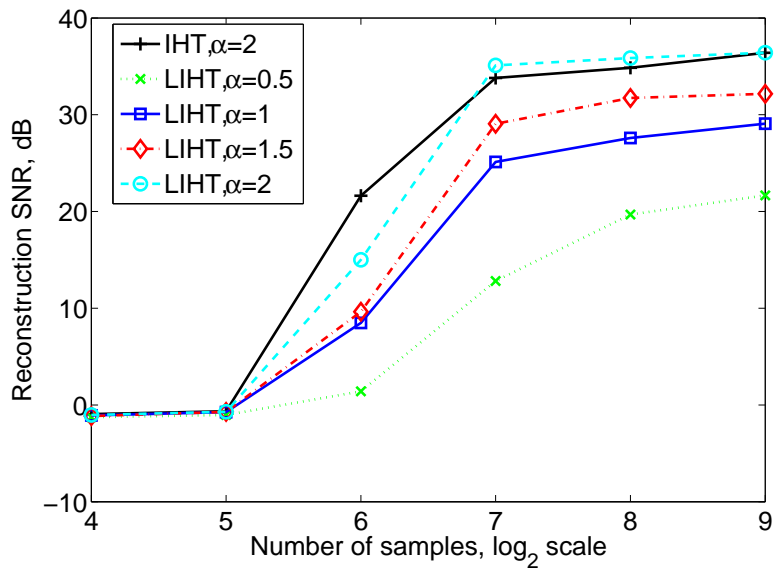}}
\caption{Reconstruction SNR as a function of the number of measurements, $s=8$ and $n=1024$.} \label{Ch6fig:2}
\end{figure}

The next experiment evaluates the computational speed of LIHT compared to the previously proposed Lorentizian BP. For this experiment we measure the reconstruction time required by the two algorithms for different signal lengths, $n=128, 256, 512, 1024, 2048$. We employ dense Gaussian sensing matrices (no fast matrix multiplication available) and fix $m=n/2$. Cauchy noise with $\sigma=0.1$ is added to the measurements. The sparsity level is fixed to $s=8$ for all signals lengths. The results are summarized in Table~\ref{Ch6tab1} with all times measured in seconds. All results are averaged over 200 realizations of the sensing matrix and the signals. The reconstruction times show that LIHT is at least three orders of magnitude faster than Lorentzian BP, with both algorithms being robust to impulsive noise. Thus, LIHT presents a fast alternative for sparse recovery in impulsive environments. One note is that the reconstruction times can be improved if structured sensing matrices that offer fast application of the sensing operator and its adjoint are used. Examples of these fast operators are the partial Fourier or Hadamard ensembles or binary sensing matrices. 
\begin{table}
\caption{Reconstruction times (in seconds) for LIHT and Lorentzian BP, $m=n/2$.} \label{Ch6tab1} \centering
\begin{tabular}{||c|c|c||}
  \hline \hline
  n & LBP & LIHT \\
  \hline \hline
  2048 & 758.0145 & 0.1755 \\
  \hline
   1024 & 116.5853 & 0.0730 \\
  \hline
   512 & 26.3145 & 0.0426 \\
  \hline
   256 & 8.7281 & 0.0102 \\
  \hline
   128 & 3.3747 & 0.0059 \\
  \hline \hline
\end{tabular}
\end{table}

The last experiment in this subsection shows the effectiveness of LIHT to recover real signals from corrupted measurements. We take random Hadamard measurements of the the $256\times 256$ ($n=65536$) Lena image and then add Cauchy distributed noise to the measurements. For all experiments we use the Daubechies Db8 wavelet transform as the sparsity basis and assume a sparsity level of $s=6000$. We fix the number of measurements as $m=32000$ and set the scale (dispersion) parameter of the  Cauchy noise to $\sigma=1$. Figure~\ref{Ch6Imex1} shows the clean measurements on the top image and the Cauchy corrupted measurements in the bottom one. 

\begin{figure}[t]
\centering{ 
\includegraphics[width = 0.7\columnwidth]{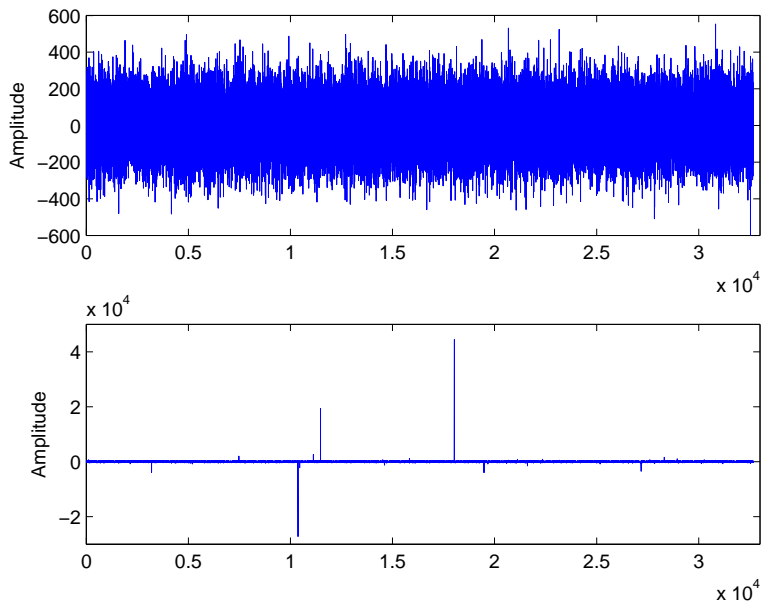}}
\caption{Example of a $256\times256$ image sampled by a random Hadamard ensemble, $m=32000$. Top: clean measurements. Bottom: Cauchy corrupted measurements, $\sigma=1$.} \label{Ch6Imex1}
\end{figure}

We compare the reconstruction results of LIHT to those obtained by the classical LS-IHT algorithm, the LS-IHT with noise clipping and LS-IHT with the measurement rejection method proposed in \cite{Laska11}. To set a clipping rule we assume that we know before hand the the range of the clean measurements and all samples are clipped within this range, \emph{i.e.}
\begin{equation*}
y_{i}^c  = \left\{ \begin{array}{rcl}
-\lambda, & \mbox{ }& \text{if}~y_i\leq -\lambda  \\ 
y_i, & \mbox{ }& \text{if}~|y_i| < \lambda \\ 
\lambda, & \mbox{ }& \text{if}~y_i\geq \lambda, 
\end{array}\right.
\end{equation*}
where $y^c$ denotes the vector of clipped measurements. For the measurement rejection approach we adapt the framework in \cite{Laska11} to address impulsive noise rather than saturation noise. We discard large measurements and form a new measurement vector as $y^r=y_{S_r}$, where $S_r=\{i | |y_i| < \lambda\}$. To find the optimal $\lambda$ for both approaches we perform an exhaustive search. Table \ref{Ch6tab2} presents the reconstruction results for different values of $\lambda$ in terms of $B$, where $B=\max_i|y_{0i}|$ and $y_0$ denotes the clean measurement vector. Thus, we select $\lambda=B$ for the clipping approach and $\lambda=0.5B$ for the measurement rejection approach. We also compare LIHT to the recovery of sparsely corrupted signals (RSCS) framework proposed in \cite{Studer12}. In this framework a sparse signal and error model is assumed and both signal and error are estimated at the same stage using an $L_1$ minimization problem with an augmented measurement matrix. In our experiments, we assume no signal/error support knowledge for RSCS. For LIHT we estimate $\gamma$ using equation \eqref{Ch6gamma}. 

\begin{table}
\caption{R-SNR (in db) for LS-IHT with clipping and rejection for different values of $\lambda$. $B=\max_i|y_{0i}|$.} \label{Ch6tab2} \centering
\begin{tabular}{||c|c|c|c|c|c|c||}
  \hline \hline
  $\lambda$ & 0.5$B$ & $B$&2$B$ &3$B$&4$B$&5$B$\\
  \hline \hline
  Clipping & -0.2 & 13.0&11.4&10.2&9.3&6.0 \\
  \hline
   Rejection & 16.2 & 15.4&14.4&13.4&12.8&10.9 \\
  \hline \hline
\end{tabular}
\end{table}

Figure~\ref{Ch6Imex2}~(a) shows the reconstructed image using  LS-IHT, R-SNR=-5.3~dB. Figure~\ref{Ch6Imex2}~(b) and \ref{Ch6Imex2}~(c) show the reconstructed images using  LS-IHT with noise clipping, R-SNR=13.0~dB, and measurement rejection, R-SNR=16.2~dB, respectively. Figure~\ref{Ch6Imex2}~(d) shows the reconstructed image by  RSCS, R-SNR=17.16~dB and Figure~\ref{Ch6Imex2}~(e) shows the reconstructed image using  LIHT, R-SNR=19.8~dB. Figure~\ref{Ch6Imex2}~(f) shows the reconstructed image from noiseless measurements using  LS-IHT as comparison, R-SNR=22.8~dB. From the results it is clear that LIHT outperform the other approaches with a reconstruction quality about 3~dB worse than the noiseless reconstruction. We also evaluate the reconstruction quality of LIHT and the benchmark methods as the number of measurements is varied. Table \ref{Ch6tab3} presents the results for four different number of measurements, $m=\{2s,3s,4s,5s\}$, where $s=6000$ is the sparsity level. The results show the advantage of robust operators in impulsive environments, especially when the number of measurements is limited. 
\begin{table}
\caption{Lena reconstruction results from Cauchy corrupted measurements.
 R-SNR (in db) as a function of $m$. $s=6000$.} \label{Ch6tab3} \centering
\begin{tabular}{||c|c|c|c|c||}
  \hline \hline
  $m$ & 2$s$ & 3$s$&4$s$ &5$s$\\
  \hline \hline
  LS-IHT & -8.5 & -5.7&-5.5&-3.4 \\
  \hline
  Clipping & 3.9 & 8.9&9.9&11.5 \\
  \hline
  Rejection & 4.7 & 10.3&11.6&14.0\\
  \hline
  RSCS  & 4.8 & 10.9&11.9&16.8 \\ 
  \hline
  LIHT & 6.9 & 12.3&13.9&17.9\\
  \hline \hline
\end{tabular}
\end{table}

\begin{figure}[p]
\centering
    
    \subfigure[]{\includegraphics[trim = 3cm 1cm 1.8cm 0.5cm, clip, keepaspectratio, width = 5.5cm]   {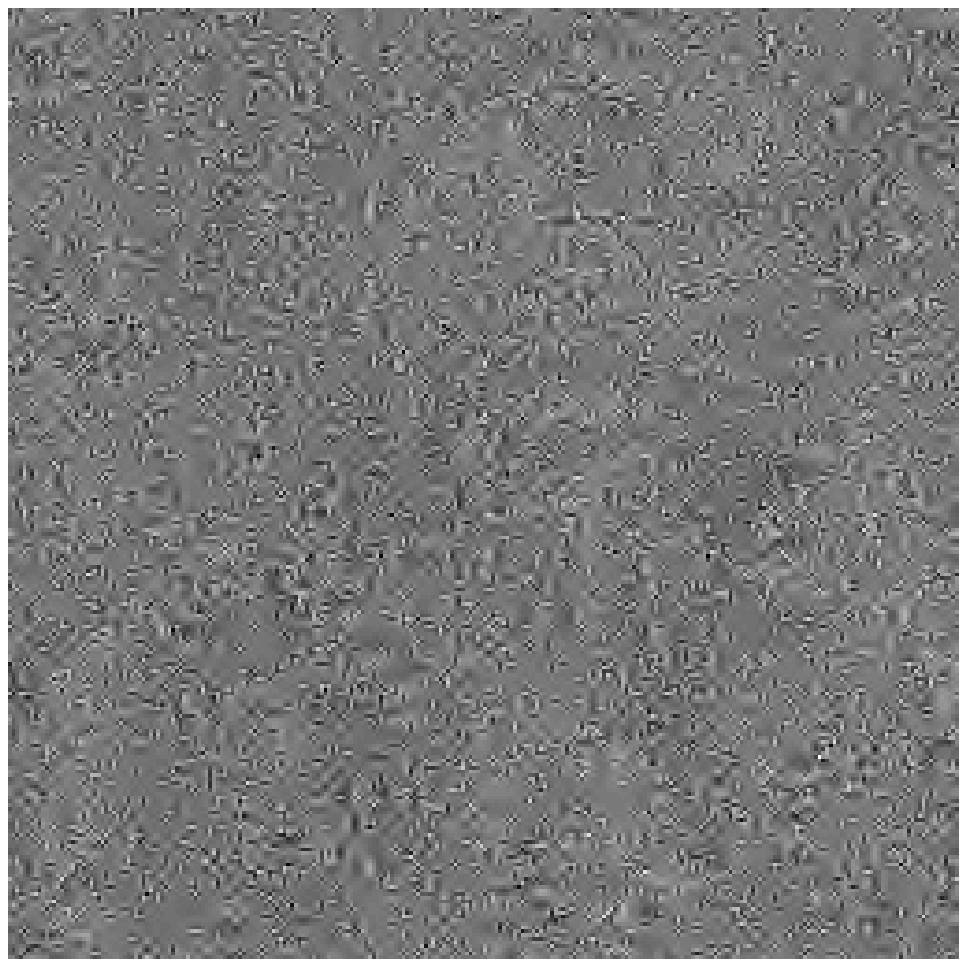}}
    \subfigure[]{\includegraphics[trim = 3cm 1cm 1.8cm 0.5cm, clip, keepaspectratio, width = 5.5cm]   {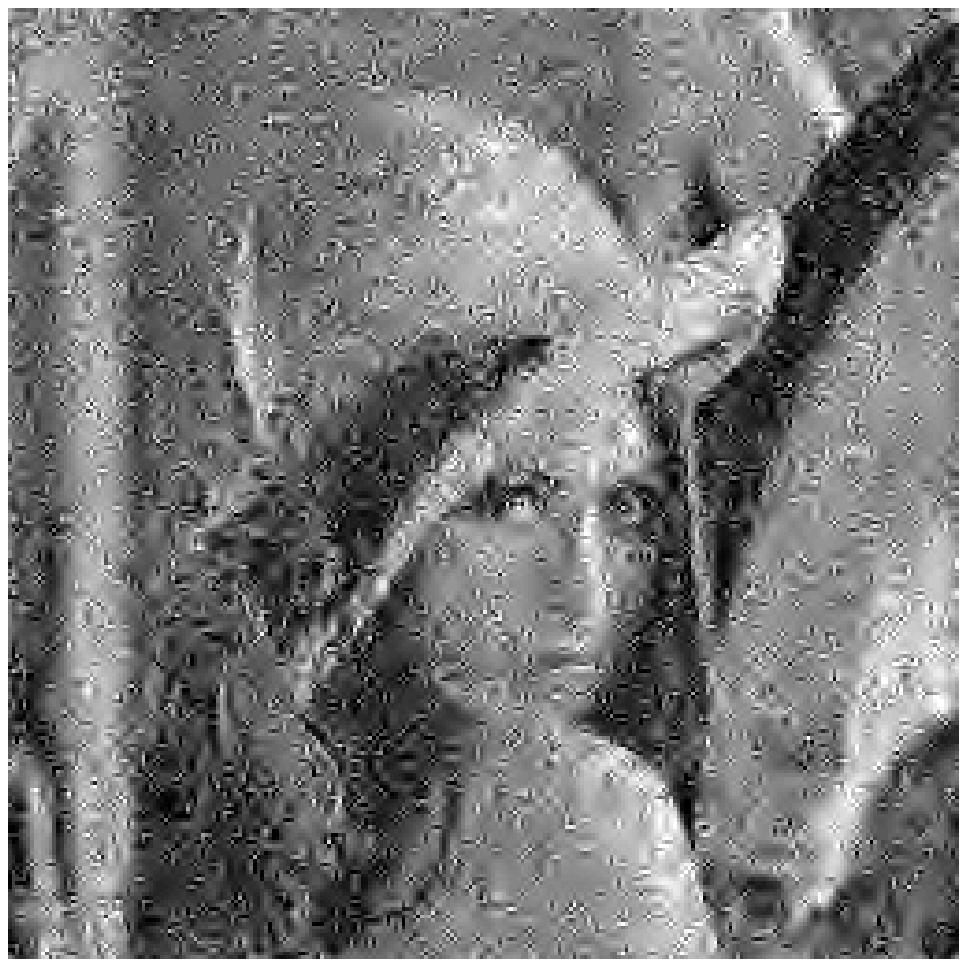}}
    \subfigure[]{\includegraphics[trim = 3cm 1cm 1.8cm 0.5cm, clip, keepaspectratio, width = 5.5cm]   {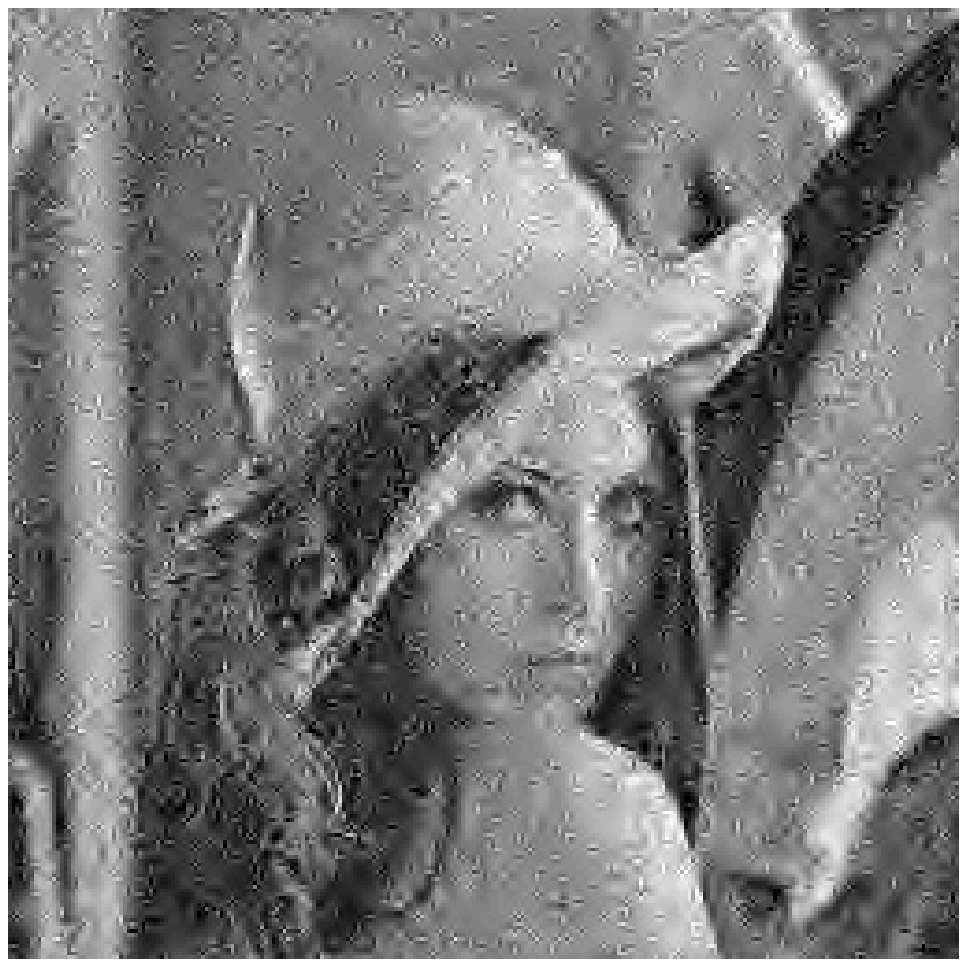}}
    \subfigure[]{\includegraphics[trim = 3cm 1cm 1.8cm 0.5cm, clip, keepaspectratio, width = 5.5cm]   {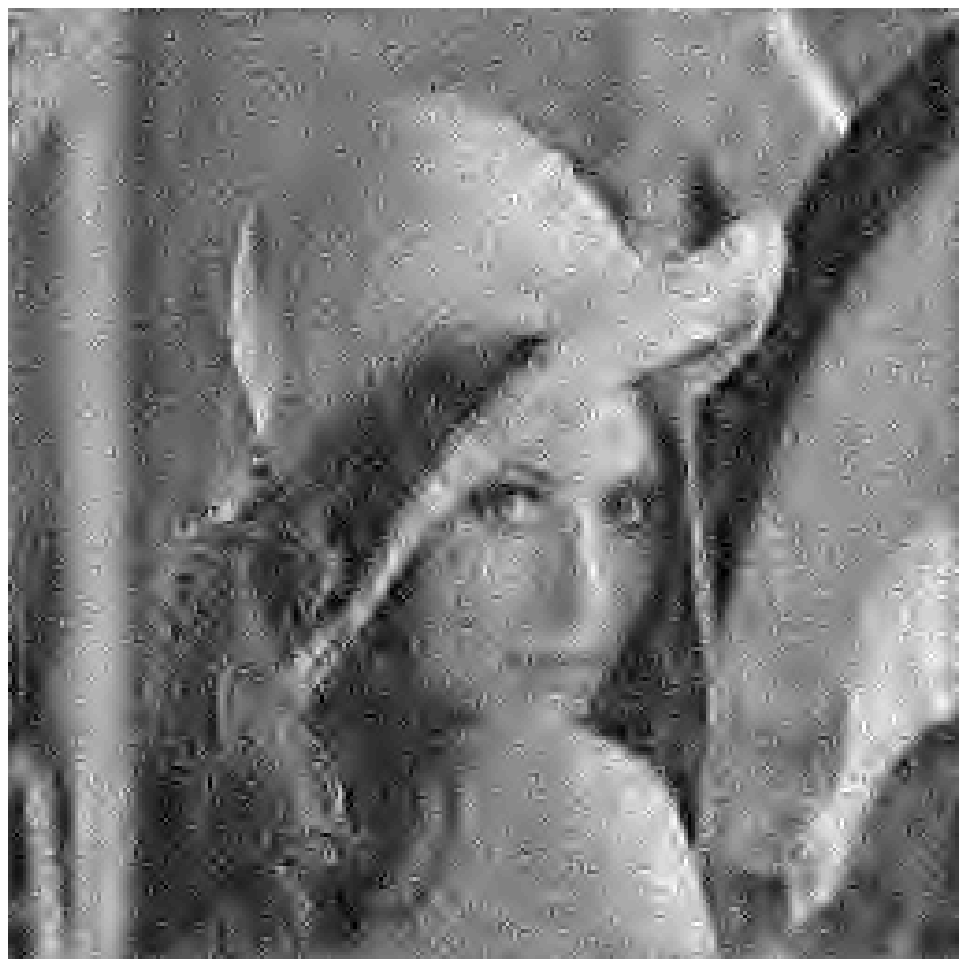}}
    \subfigure[]{\includegraphics[trim = 3cm 1cm 1.8cm 0.5cm, clip, keepaspectratio, width = 5.5cm]   {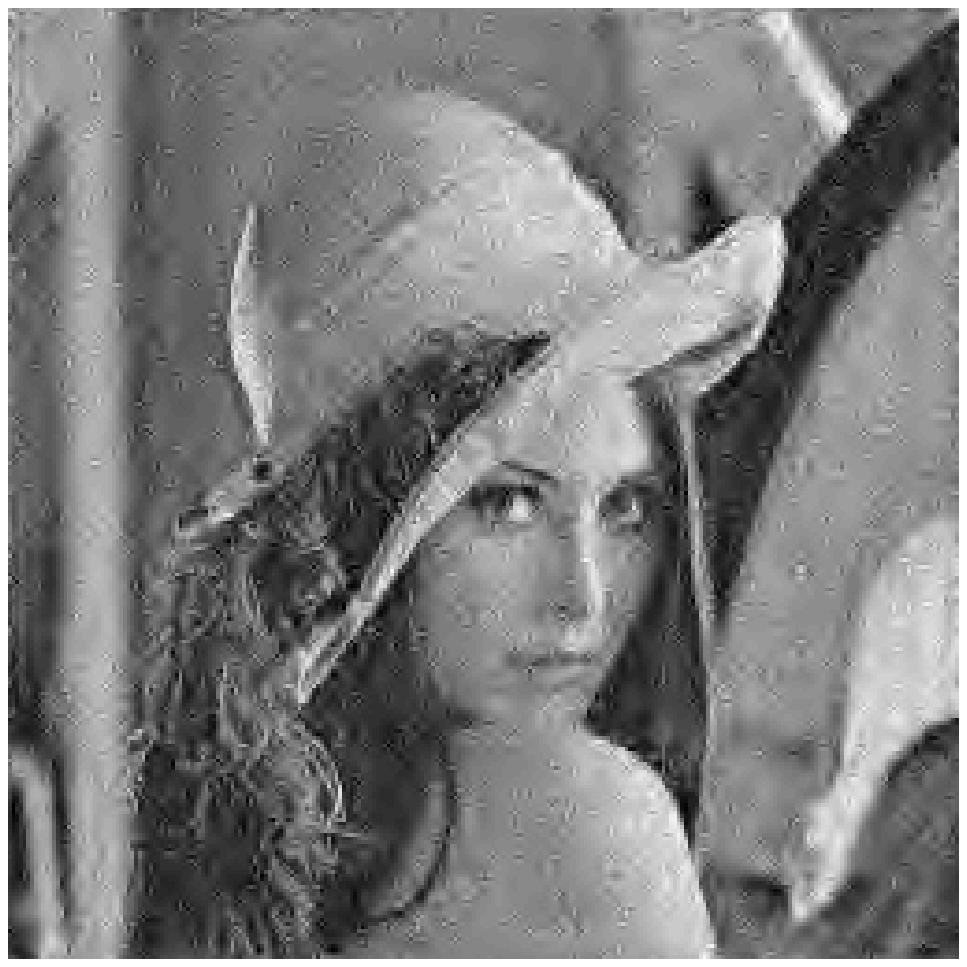}}
    \subfigure[]{\includegraphics[trim = 3cm 1cm 1.8cm 0.5cm, clip, keepaspectratio, width = 5.5cm]   {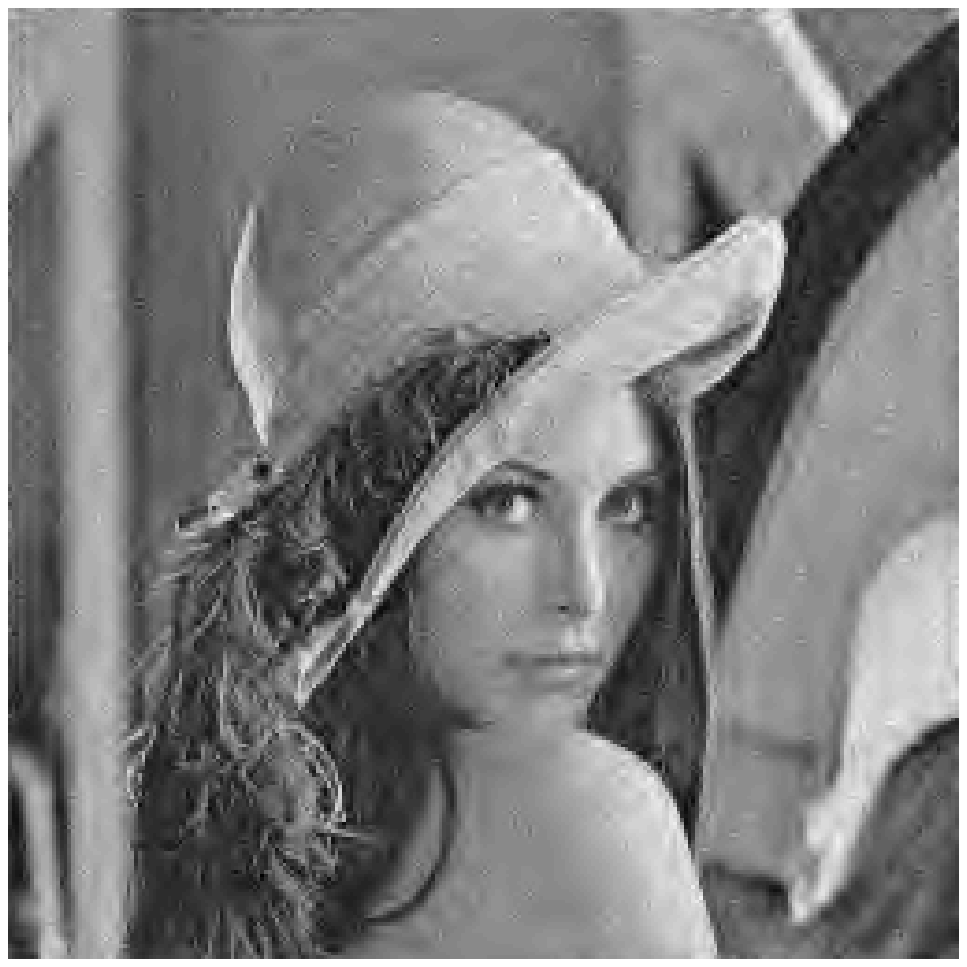}}
    
\caption{Lena image reconstruction example from measurements corrupted by Cauchy noise, $m=32000$ and $s=6000$.  (a) Reconstructed image using  LS-IHT, R-SNR=-5.3~dB. (b) Reconstructed image using  LS-IHT and noise clipping, R-SNR=13.0~dB. (c) Reconstructed image using  LS-IHT and measurement rejection, R-SNR=16.2~dB. (d) Reconstructed image using  RSCS, R-SNR=17.2~dB. (e) Reconstructed image using  LIHT, R-SNR=19.8~dB. (f) Reconstructed image from noiseless measurements using  LS-IHT, R-SNR=22.8~dB.}
\label{Ch6Imex2}%
\end{figure}

\subsection{LIHT with Partially Known Support}
Numerical experiments that illustrate the effectiveness of LIHT with partially known support are presented in this section. Results are presented for synthetic and real signals. In the real signal case, comparisons are made with a broad set of alternative algorithms.

Synthetic sparse vectors are employed in the first experiment. The signal length is set as $n=1000$ and the sparsity level is fixed to $50$. The nonzero coefficients are drawn from a Rademacher distribution, their position randomly chosen and amplitudes $\{-10,10\}$. The vectors are sampled using sensing matrices $\Phi$ that have i.i.d. entries drawn from a standard normal distribution with normalized columns. Each experiment is repeated 300 times, with average results presented.

The effect of including partial support knowledge is analyzed by increasing the cardinality of the known set in steps of $10\%$  for different numbers of measurements. The probability of exact reconstruction is employed as a measure of performance. Figure~\ref{Ch6fig:PKS} shows that, as expected, the reconstruction accuracy grows with the percentage of known support. The results also show that incorporating prior support information substantially reduces the number of measurements required for successful recovery.
\begin{figure}[t]
\centering{ 
\includegraphics[width = 0.7\columnwidth]{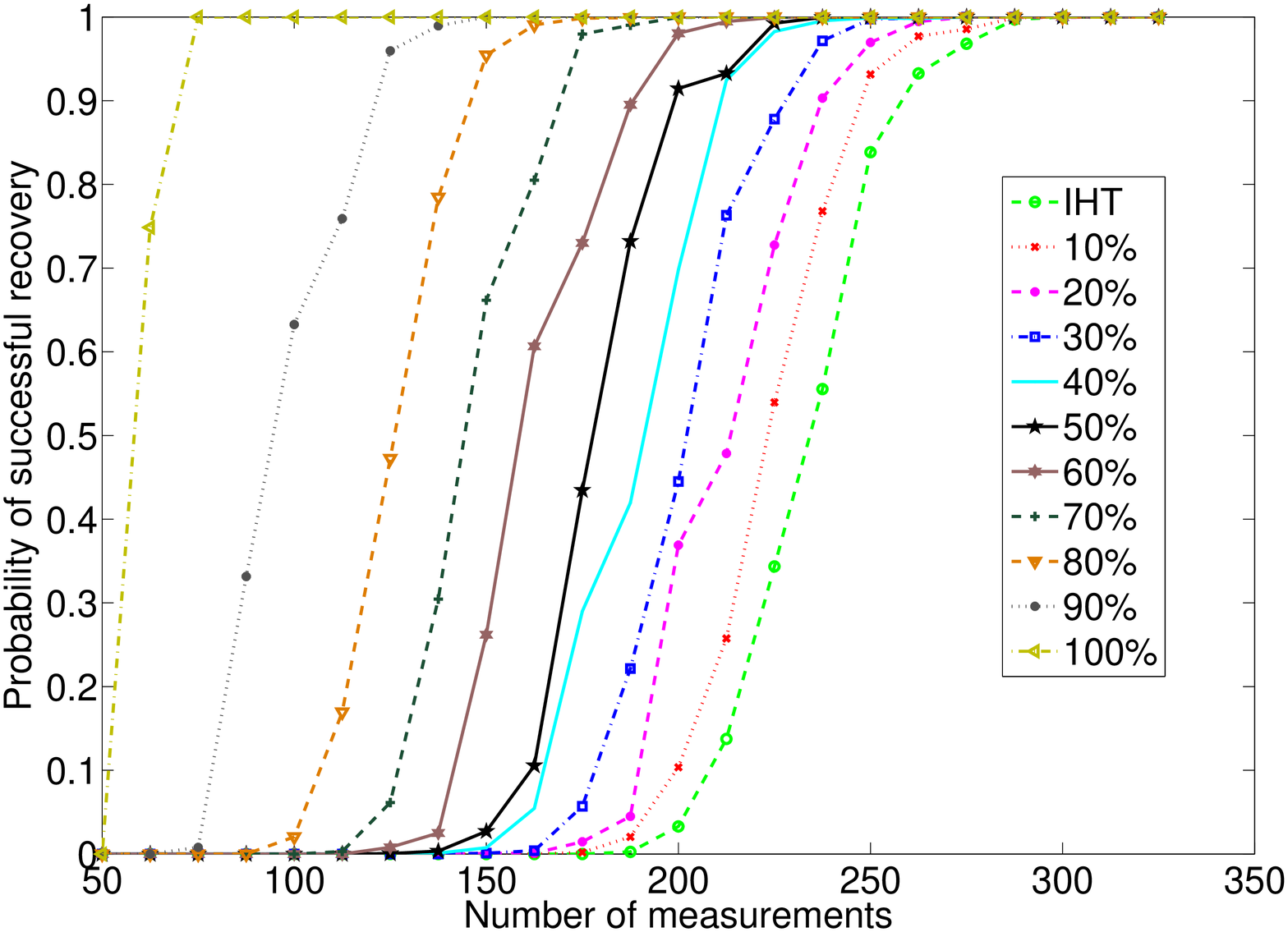}}
\caption{Probability of successful recovery as a function of the number of measurements, for different percentages of partially known support and signal length $n=1000$.} \label{Ch6fig:PKS}
\end{figure}

The second experiment illustrates algorithm performance for real compressible signals. ECG signals are utilized due to the structure of their sparse decompositions. Experiments are carried out over 10-min long leads extracted from records 100, 101, 102, 103, 107, 109,
111, 115, 117, 118 and 119 from the MIT-BIH Arrhythmia Database (see~\cite{Blanco08} and references therein). Cosine modulated filter banks are used to determine a sparse representation of the signal~\cite{Blanco08}. A sparse signal approximation is determined by processing 1024 samples of ECG data, setting the number of channels, $M$, to 16, and selecting the largest 128 coefficients. This support set is denoted by $T$; note that $|T|=128$. Figure~\ref{Ch6fig:ECG2} shows an example of a decomposition of a lead of 1024 samples and its decomposition using CMFB.

\begin{figure}[t]
\centering{ 
\includegraphics[width = 0.7\columnwidth]{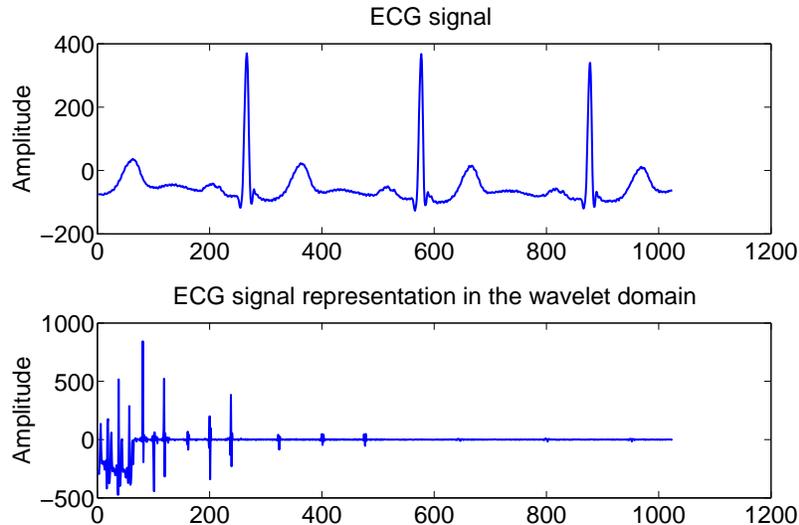}}
\caption{Decomposition of an ECG signal using CMFB, $M=16$ and $n=1024$.} \label{Ch6fig:ECG2}
\end{figure}

Three cases are considered. In the first, the median (magnitude) support coefficient is determined and the coefficients of $T$ with magnitudes greater than or equal to the median are designated as the known signal support, \emph{i.e.}, the positions of the largest (magnitude) 50\% of $T$ coefficients are taken to be the known signal support. This case is denoted as IHT-PKS-I. The second partially known support case corresponds to those with magnitude less than the median, \emph{i.e.}, the positions of the smallest (magnitude) 50\% of $T$ coefficients since these  might be the most difficult to find coefficients. This case is denoted as IHT-PKS-II. The third and final selection, denoted as IHT-PKS, is related to the low-pass approximation of the first subband, which corresponds to the first 64 coefficients (when $n=1024$). This first subband accumulates the majority of signal energy, which is the motivation for this case.

Figure~\ref{Ch6fig:ECG} compares the three proposed partially known support selections.  Each method improves the performance over standard LIHT, except for IHT-PKS-II when the number of measurements is not sufficient to achieve accurate reconstruction. Note, however, that the performance of IHT-PKS-II improves rapidly as the number of measurements increases, with the method outperforming the other algorithms in this regime. The performance of  IHT-PKS-I is very similar to IHT-PKS since most of the first subband low-pass approximation coefficients are included in the $50\%$ largest coefficients of $T$ set. Notice that IHT-PKS-I performs slightly better than IHT-PKS for small numbers of measurements.

\begin{figure}[t]
\centering{ 
\includegraphics[width = 0.7\columnwidth]{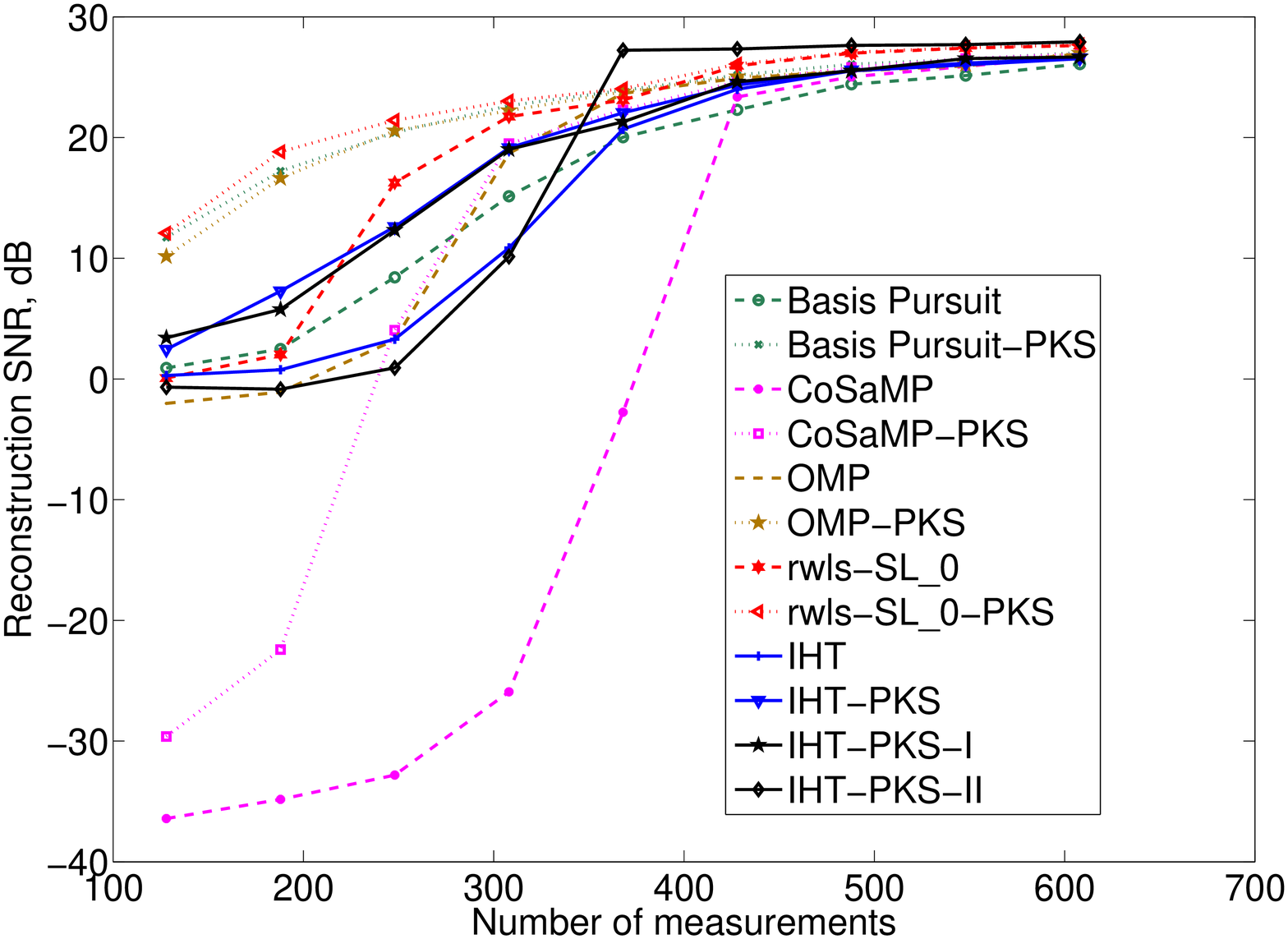}}
\caption{Comparison of LIHT, BP, OMP, CoSaMP, rwls-S$L_0$ and their partially known support versions for ECG signals of length $n=1024$.} \label{Ch6fig:ECG}
\end{figure}

Also compared with LIHT in Figure~\ref{Ch6fig:ECG} are the OMP, CoSaMP, and rwls-S$L_0$ iterative algorithms,
as well as their partially known support versions (OMP-PKS, CoSaMP-PKS, and rwls-S$L_0$-PKS)~\cite{Carr10c}. For reference, we also include Basis Pursuit (BP) and Basis Pursuit with partially known support (BP-PKS)~\cite{Vaswani09}. In all cases, the positions of the first subband low-pass approximation coefficients are selected as the signal partially known support. Note that LIHT-PKS performs better than CoSaMP-PKS for small numbers of measurements and yields similar reconstructions when the number of measurements increases. Although the known support versions of the other iterative algorithms require fewer measurements to achieve accurate reconstructions, LIHT does not require the exact solution to an inverse problem, thus making it computationally more efficient. And as in the previous example, the performance of Lorentzian iterative hard thresholding is improved through the inclusion of partially known support information, thereby enabling the number of measurements requires for a specified level of performance to be reduced.

\begin{figure}[t]
\centering{ 
\includegraphics[width = 0.6\columnwidth]{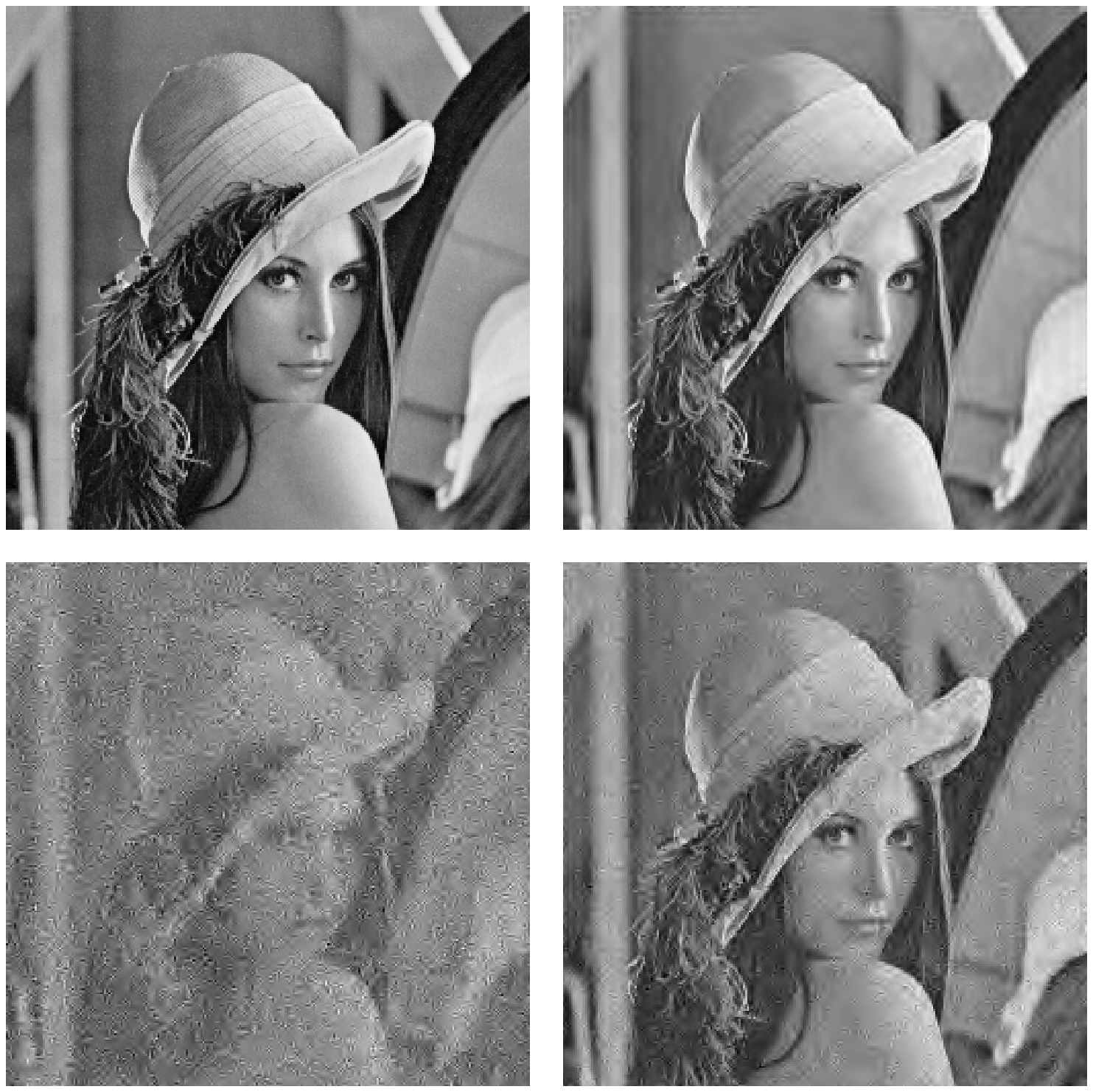}}
\caption{Top left: Original 256$\times$256 image. Top right: Best $s$-term approximation, $s=6000$, R-SNR=23.9~dB. Reconstruction from $m=16000$ measurements. Bottom left: LIHT, R-SNR=10.2~dB. Bottom right: LIHT-PKS $k=2048$, R-SNR=20.4~dB.} \label{Ch6fig:Im2}
\end{figure}
As a final example we illustrate how the partially known support framework can be applied in image reconstruction. Consider a wavelet decomposition of natural images. It is observed that the largest coefficients are concentrated in the approximation band and the remainder signal, detail coefficients, is a sparser signal than the original decomposition. Thus, a possible form to incorporate the partially known support framework is to assume that the approximation band coefficients are part of the true signal support, \emph{i.e.}, the partially known support. To test our assumption we take random Hadamard measurements of the the $256\times 256$ Lena image and then we estimate the image from the measurements. Figure~\ref{Ch6fig:Im2} top left shows the original image. We use the Daubechies DB8 wavelet transform as our sparsity basis and we approximate the image with the largest $6000$ coefficients, thus $|T|=6000$. Figure~\ref{Ch6fig:Im2} top right shows the best $s$-term approximation, $s=6000$, with R-SNR=23.9~dB for comparison. We take $m=16000$ measurements and reconstruct the image using the LIHT algorithm and the LIHT-PKS algorithm. For LIHT-PKS we assume that the approximation band is in the true support of the image coefficients, $k=2048$ for this example. The reconstruction results are shown in Figure~\ref{Ch6fig:Im2} bottom left and Figure~\ref{Ch6fig:Im2} bottom right, respectively. The reconstruction SNRs are R-SNR=10.2~dB for the standard LIHT and R-SNR=20.4~dB for LIHT-PKS. The LIHT-PKS algorithm outperforms its counterpart without support knowledege by 10~dB, but more importantly, the partially known support reconstruction quality is 3~dB below the reconstruction quality obtained by the best $s$-term approximation.

\section{Concluding Remarks}
\label{sec:Ch6conc}
This paper presents a Lorentzian based IHT algorithm for recovery of sparse signals in impulsive environments. The derived algorithm is comparable to least squares based IHT in terms of computational load, with the advantage of robustness against heavy-tailed impulsive noise. Sufficient conditions for stability are studied and a reconstruction error bound is derived that depends on the noise strength and a tunable parameter of the Lorentzian norm. Simulations results show that the LIHT algorithm yields comparable performance with state of the art algorithms in light-tailed environments while having substantial performance improvements in heavy-tailed environments. Simulation results also show that LIHT is a fast reconstruction algorithm with scalability for large dimensional problems. Methods to estimate the adjustable parameters in the reconstruction algorithm are proposed, although computation of their optimal values remains an open question. Future work will focus on convergence analysis of the proposed algorithm. 

Additionally, this paper proposes a modification of the LIHT algorithm that incorporates known support in the recovery process. Sufficient conditions for stable recovery in the compressed sensing with partially known support problem are derived. The theoretical analysis shows that including prior support information relaxes the conditions for successful reconstruction. Numerical results show that the modified LIHT improves performance, thereby requiring fewer samples to yield an approximate reconstruction.

\bibliographystyle{IEEEtran}
\bibliography{abrev,RGCD,carrillo-publications,RCS,RGCD1,robustlmestimation}

\end{document}